\documentclass[11pt]{article}

\usepackage{graphicx}
\usepackage{amsfonts}
\usepackage{amsmath}
\usepackage{amsthm}
\usepackage[mathcal]{euscript}      % Define tensors
\usepackage{bm}                     % Define vectors
\usepackage{color}
\usepackage{epstopdf}
\usepackage{amssymb}
\usepackage{mathrsfs,mathdots}
\usepackage{lscape}
\usepackage[version=3]{mhchem}

\usepackage[paper=a4paper,dvips,top=3cm,left=2.4cm,right=2.4cm,
    foot=1cm,bottom=3cm]{geometry}

\begin{document}

\title{An Irreducible Function Basis of Isotropic Invariants of A Third Order Three-Dimensional Symmetric Tensor}
\author{Zhongming Chen\footnote{%
    Department of Mathematics, School of Science, Hangzhou Dianzi University, Hangzhou 310018, China ({\tt zmchen@hdu.edu.cn}).
    This author's work was supported by the National Natural Science Foundation of China (Grant No. 11701132).}
\and Jinjie Liu\footnote{%
    Department of Applied Mathematics, The Hong Kong Polytechnic University,
    Hung Hom, Kowloon, Hong Kong ({\tt jinjie.liu@connect.polyu.hk}).
    }
\and Liqun Qi\footnote{%
    Department of Applied Mathematics, The Hong Kong Polytechnic University,
    Hung Hom, Kowloon, Hong Kong ({\tt maqilq@polyu.edu.hk}).
    This author's work was partially supported by the Hong Kong Research Grant Council
    (Grant No. PolyU  15302114, 15300715, 15301716 and 15300717).}
\and Quanshui Zheng\footnote{%
    Department of Engineering Mechanics, Tsinghua University, Beijing 100084, China ({\tt zhengqs@tsinghua.edu.cn}).}
\and Wennan Zou\footnote{%
    Institute for Advanced Study, Nanchang University, Nanchang 330031, China ({\tt zouwn@ncu.edu.cn}).
    This author was supported by the National Natural Science Foundation of China (Grant No. 11372124)}
    }

\date{June 22, 2018}
\maketitle

\begin{abstract}    In this paper, we present an eleven invariant isotropic irreducible function basis of a third order three-dimensional symmetric tensor.  This irreducible function basis is a proper subset of the Olive-Auffray minimal isotropic integrity basis of that tensor.   The octic invariant and a sextic invariant in the Olive-Auffray integrity basis are dropped out.   This result is of significance to the further research of irreducible function bases of higher order tensors.

\end{abstract}

\newtheorem{Theorem}{Theorem}[section]
\newtheorem{Definition}[Theorem]{Definition}
\newtheorem{Lemma}[Theorem]{Lemma}
\newtheorem{Corollary}[Theorem]{Corollary}
\newtheorem{Proposition}[Theorem]{Proposition}
\newtheorem{Conjecture}[Theorem]{Conjecture}
\newtheorem{Question}[Theorem]{Question}

% LaTeX definitions
\renewcommand{\hat}[1]{\widehat{#1}}
\renewcommand{\tilde}[1]{\widetilde{#1}}
\renewcommand{\bar}[1]{\overline{#1}}
\newcommand{\REAL}{\mathbb{R}}
\newcommand{\COMPLEX}{\mathbb{C}}
\newcommand{\SPHERE}{\mathbb{S}^2}
\newcommand{\diff}{\,\mathrm{d}}
\newcommand{\st}{\mathrm{s.t.}}
\newcommand{\T}{\top}
\newcommand{\vt}[1]{{\bf #1}}%{\bm{#1}}
\newcommand{\aaa}{{\vt{a}}}
\newcommand{\ddd}{{\vt{d}}}
\newcommand{\x}{{\vt{x}}}
\newcommand{\y}{{\vt{y}}}
\newcommand{\z}{{\vt{z}}}
\newcommand{\uu}{{\vt{u}}}
\newcommand{\vv}{{\vt{v}}}
\newcommand{\ww}{{\vt{w}}}
\newcommand{\e}{{\vt{e}}}
\newcommand{\g}{{\vt{g}}}
\newcommand{\0}{{\vt{0}}}
\newcommand{\Ten}{\bf{T}}
\newcommand{\HH}{\mathbb{H}}
\newcommand{\A}{{\bf A}}
\newcommand{\B}{{\bf B}}
\newcommand{\C}{\mathcal{C}}
\newcommand{\D}{{\bf D}}
\newcommand{\E}{\bf{E}}
\newcommand{\OOO}{\mathcal{O}}
\newcommand{\U}{{\bf{U}}}
\newcommand{\V}{{\bf{V}}}
\newcommand{\W}{{\bf{W}}}
\newcommand{\I}{{\bf{I}}}
\newcommand{\II}{{\mathcal{I}}}
\newcommand{\OO}{{\bf{O}}}
\newcommand{\RESULTANT}{\mathrm{Res}}

\newpage
\section{Introduction}

Tensor function representation theory constitutes an important fundamental of theoretical and applied mechanics. Representations of complete and irreducible basis for isotropic
invariants could predict the available nonlinear constitutive theories by the formulation of energy term. Since irreducible representations for tensor-valued functions can be immediately yielded from known irreducible representations for invariants (scalar-valued functions) \cite{Zh-94}, the studies of isotropic function basis have most priority.  Perhaps we may trace back the modern development of tensor representation theory to the great mathematician Hermann Weyl's book \cite{We-16}. This book was first published in 1939. Here, we cite its new edition in 2016.     Then, since the 1955 paper of Rivilin and Ericksen \cite{RE-55}, many researchers, such as Smith, Pipkin, Spencer, Boehler, Betten, Pennisi and Zheng \cite{SR-58, PR-59, Boe-87, PT-87, Spe-70, ZB-95, Zh-96}, to name only a few of them here, have made important contributions to tensor representation theory.    For the literature of tensor representation theory before 1994, people may find it in the 1994 survey paper of Zheng \cite{Zh-94}.   The development of tensor representation theory after 1994 paid more attentions to
minimal integrity bases of isotropic invariants of third and fourth order three-dimensional tensors \cite{BKO-94, OA-14, OKA-17, SB-97, Zh-96, ZB-95}.  The polynomial basis of anisotropic invariants of the elasticity tensor was studied by Boehler, Kirillov and Onat \cite{BKO-94} in 1994. %Here we focus on %the function basis of isotropic invariants of \textcolor{blue}{third} order tensors.
Zheng and Betten \cite{ZB-95} and Zheng \cite{Zh-96} studied the tensor function representations involving tensors of
orders higher than two.  Smith and Bao \cite{SB-97} presented minimal integrity bases of isotropic invariants for third and fourth order three-dimensional symmetric and traceless tensors in 1997.
Note that Boehler, Kirillov and Onat \cite{BKO-94} had already given a minimal integrity basis for a fourth order three-dimensional symmetric and traceless tensor in 1994.  But the minimal integrity basis given by Smith and Bao \cite{SB-97} for the same tensor is slightly different \cite{CQZ-18}.
In 2014, an integrity basis with thirteen isotropic invariants of a (completely) symmetric third order three-dimensional tensor was presented by Olive and Auffray \cite{OA-14}.  Olive \cite{Ol-17}[p. 1409] stated that this integrity basis is a minimal integrity basis.   Olive, Kolev and Auffray \cite{OKA-17} presented a minimal integrity basis of the elasticity tensor, with $297$ invariants, in 2017.   Very recently, Chen, Qi and Zou \cite{CQZ-18} showed that the four invariant Smith-Bao minimal isotropic integrity basis of a third order three-dimensional symmetric and traceless tensor is also an irreducible function basis of that tensor.

In this paper, the summation convention is used.    If an index is repeated twice in a product, then it means that this product is summed up with respect to this index from $1$ to $3$.

Suppose that a tensor $\A$ has the form $A_{i_1\dots i_m}$ under an orthonormal basis $\{ \e_i \}$.  A scalar function of $\A$, $f(\A) = f(A_{i_1\dots i_m})$ is said to be an isotropic invariant of $\A$ if for any orthogonal matrix $q_{ij}$, we have
$$f(A_{i_1\dots i_m}) = f(A_{j_1\dots j_m}q_{i_1j_1}\dots q_{i_mj_m}).$$
%\textcolor{red}{Consider a set of isotropic invariants $f_1, \ldots, f_r$ of $\A$. Any single-valued function or polynomial of $f_1, \ldots, f_r$:
%\begin{equation} \label{e0}
%\phi(\A) = \phi(f_1, \ldots, f_r)
%\end{equation}
%is called a representation for isotropic scalar-valued functions or polynomials of $\A$. The representation is said to be \emph{irreducible}, if none of $f_1, \ldots, f_r$ is expressible as %a single-valued function or polynomial of the remainders; to be \emph{complete}, if any isotropic scalar-valued function or polynomial of $\A$ can be expressed in the form (\ref{e0}).}
  By \cite{Zh-94, ZB-95}, a set of isotropic polynomial invariants $f_1, \dots, f_r$ of $\A$ is said to be an \emph{integrity basis} of $\A$ if any isotropic polynomial invariant is a polynomial of $f_1, \dots, f_r$, and a set of isotropic invariants $f_1, \dots, f_m$ of $\A$ is said to be a \emph{function basis} of $\A$ if any isotropic invariant is a function of $f_1, \dots, f_m$.  An integrity basis is always a function basis but not vice versa \cite{OKA-17}.    A set of isotropic polynomial invariants $f_1, \dots, f_r$ of $\A$ is said to be \emph{polynomially irreducible} if none of them can be a polynomial of the others. Similarly, a set of isotropic invariants $f_1, \dots, f_m$ of $\A$ is said to be \emph{functionally irreducible} if none of them can be a function of the others.   An integrity basis of $\A$ is said to be a \emph{minimal integrity basis} of $\A$ if it is polynomially irreducible, and a function basis of $\A$ is said to be an \emph{irreducible function basis} of $\A$ if it is functionally irreducible.

In this paper, we present an eleven invariant isotropic irreducible function basis of a third order three-dimensional symmetric tensor.  This irreducible function  basis is a proper subset of the Olive-Auffray minimal isotropic integrity basis of that tensor.   The octic invariant and a sextic invariant in the Olive-Auffray integrity basis are dropped out.

In the next section, some preliminary results are given.  These include the minimal integrity basis result of Smith and Bao \cite{SB-97} for a third order three-dimensional symmetric and traceless tensor, the consequent result of Chen, Qi and Zou \cite{CQZ-18} to confirm it is also an irreducible function basis, and the result of Olive and Auffray \cite{OA-14} for a minimal integrity basis of a third order three-dimensional symmetric tensor.

In Section 3, we present an eleven invariant isotropic function basis of a third order three-dimensional symmetric tensor.  This function basis is obtained by using two syzygy relations to drop out the octic invariant and a sextic invariant from the Olive-Auffray integrity basis. Note that \emph{a syzygy relation} is a set of coefficients in the polynomial ring such that the corresponding element generated by the function basis vanishes in the module.

Then in Section 4, we show that this function basis is indeed an irreducible function basis of a third order three-dimensional symmetric tensor.

This result is significant to the further research of irreducible function bases of higher order tensors.
First, this is the first time to give an irreducible function basis of isotropic invariants of a third order three-dimensional symmetric tensor.
%the research for an irreducible function basis of a third order symmetric tensor is vacant before. }
Second, there are still three syzygy relations among these eleven invariants.   This shows that an irreducible function basis consisting of polynomial invariants may not be algebraically minimal.
We discuss this in Section 5.

From now on, we use $\A$ to denote a third order three-dimensional tensor, and assume that it is represented by $A_{ijk}$ under an orthonormal basis $\{ \e_i \}$.
We consider the three-dimensional physical space.  Hence $i, j, k \in \{1, 2, 3\}$.   We say that $\A$ is a symmetric tensor if for $i, j, k = 1, 2, 3$, we have
$$A_{ijk} = A_{jik} = A_{ikj} = A_{kji}.$$
We say that $\A$ is traceless if
$$A_{iij} = A_{ijj} = A_{iji} = 0.$$
We use $\0$ to denote the zero vector and $\OOO$ to denote the third order three-dimensional zero tensor.

\section{Preliminaries}

In this section, we review the minimal integrity basis result of Smith and Bao \cite{SB-97} for a third order three-dimensional symmetric and traceless tensor, the consequent result of Chen, Qi and Zou \cite{CQZ-18} to confirm it is also an irreducible function basis, and the minimal integrity basis result of Olive and Auffray \cite{OA-14} for a third order three-dimensional symmetric tensor.

\subsection{An Irreducible Function Basis of A Third Order Three-Dimensional Symmetric and Traceless Tensor}

In 1997, Smith and Bao \cite{SB-97} presented the following theorem.

\begin{Theorem} \label{t2.1}
Let $\D$ be an irreducible (i.e., symmetric and traceless) third order three-dimensional tensor.   Denote $v_p:=D_{ijk}D_{ij\ell}D_{k\ell p}$,
$I_2 := D_{ijk}D_{ijk}$,  $I_4 := D_{ijk}D_{ij\ell}D_{pqk}D_{pq\ell}$, $I_6 := v_iv_i$ and  $I_{10} := D_{ijk}v_iv_jv_k$.
Then $\{ I_2, I_4, I_6, I_{10} \}$ is a minimal integrity basis of $\D$.
\end{Theorem}

Very recently, Chen, Qi and Zou \cite{CQZ-18} proved the following theorem.

\begin{Theorem} \label{t2.2}
Under the notation of Theorem \ref{t2.1},
the Smith-Bao minimal integrity basis $\{ I_2, I_4, I_6, I_{10} \}$ is also an irreducible function basis of $\D$.
\end{Theorem}

\subsection{The Olive-Auffray Integrity Basis of A Third Order Three-Dimensional Symmetric Tensor}

According to \cite{Zh-96}, we decompose a third order three-dimensional symmetric tensor $\A$ into a third order three-dimensional symmetric and traceless tensor $\D$ and a vector $\uu$, with
\begin{equation*}
  u_i=A_{i\ell\ell}
\end{equation*}
and
\begin{equation*}
  D_{ijk} = A_{ijk}-\frac{1}{5}\left(u_k\delta_{ij}+u_j\delta_{ik}+u_i\delta_{jk}\right),
\end{equation*}
where $\delta_{pq}=1$ if $p=q$ and $\delta_{pq}=0$ if $p\neq q$.

In 2014, Olive and Auffray \cite{OA-14} presented the following theorem.

\begin{Theorem} \label{t2.3}
Let $\A$ be a third order three-dimensional symmetric tensor with the above decomposition.
The following thirteen invariants
\begin{equation*}
\begin{array}{lll}
  I_2 := D_{ijk}D_{ijk}, & J_2 := u_iu_i, & I_4 := D_{ijk}D_{ij\ell}D_{pqk}D_{pq\ell}, \\
  J_4 := D_{ijk}u_kD_{ij\ell}u_{\ell}, & K_4 := D_{ijk}D_{ij\ell}D_{k\ell p}u_p, & L_4 := D_{ijk}u_ku_ju_i, \\
  I_6 := v_iv_i, & J_6 := D_{ijk}D_{ij\ell}u_kD_{\ell pq}u_pu_q, & K_6 := v_kw_k, \\
  L_6 := D_{ijk}D_{ij\ell}u_kv_{\ell}, & M_6 := D_{ijk}D_{pqk}u_iu_ju_pu_q, & I_8 := D_{ijk}D_{ij\ell}u_kD_{pq\ell}D_{pqr}v_r, \\
  I_{10} := D_{ijk}v_iv_jv_k, &&
\end{array}
\end{equation*}
where $v_p:=D_{ijk}D_{ij\ell}D_{k\ell p}$ and $w_k:=D_{ijk}u_iu_j$, form an integrity basis of $\A$.
\end{Theorem}

As an integrity basis is always a function basis, we may start from the Olive-Auffray integrity basis
$$\{I_2,J_2,I_4,J_4,K_4,L_4,I_6,J_6,K_6,L_6,M_6,I_8,I_{10}\}$$
to find an irreducible function basis of $\A$.

\section{An Eleven Invariant Function Basis}

 In this section, we show that the following eleven invariant set
 $$\{I_2,J_2,I_4,J_4,K_4,L_4,I_6,J_6,L_6,M_6,I_{10}\}$$
 is a function basis of the third order three-dimensional symmetric tensor $\A$.  Note that this set is obtained by dropping $K_6$ and $I_8$ from the Olive-Auffray integrity basis
   $$\{I_2,J_2,I_4,J_4,K_4,L_4,I_6,J_6,K_6,L_6,M_6,I_8,I_{10}\}.$$
 Thus, the task of this section is to show that $K_6$ and $I_8$ can be dropped out for a function basis.

 We first prove the following proposition.

 \begin{Proposition} \label{p3.1}
 In the Olive-Auffray integrity basis, we have
 $$2I_2J_2 - 3 J_4 \ge 0,$$
 where equality holds if and only if either $\D = \OOO$ or $\uu=\0$.
 \end{Proposition}
 \begin{proof}
 By definition, if either $\D = \OOO$ or $\uu=\0$, we have $I_2J_2 = 0$ and $J_4 = 0$.  Hence $2I_2J_2 - 3 J_4 = 0$ in this case.

Consider the optimization problem
$$\min \{ 2I_2J_2 - 3 J_4 : D_{ijk}D_{ijk}=1, u_iu_i=1.\},$$
where the variables are the seven independent components of $\D$ and the three components of $\uu$.     Using GloptiPoly 3 \cite{HLL-09} and SeDuMi \cite{St-06}, we compute  the minimum value of this optimization problem is $0.2$, where the minimizer is $D_{111} = 0.2829, D_{112}=D_{113}=0, D_{122} = -0.2828$, $D_{123}=-0.2450$, $D_{222}=0$, $D_{223}=-0.2828$, $u_1 = -0.4471$, $u_2 = -0.7746$, $u_3=-0.4474$.   Hence, the minimum value is positive.  This implies that if $2I_2J_2 - 3 J_4 = 0$ then either $\D = \OOO$ or $\uu = \0$.
 \end{proof}

We are now ready to prove the following theorem.

\begin{Theorem} \label{t3.1}
 The eleven invariant set $\{I_2,J_2,I_4,J_4,K_4,L_4,I_6,J_6,L_6,M_6,I_{10}\}$ is a function basis of the third order three-dimensional symmetric tensor $\A$.
\end{Theorem}
\begin{proof}
Consider all possible tenth degree powers and products of these thirteen invariants\\ $I_2,J_2,I_4,J_4,K_4,L_4,I_6,J_6,K_6,L_6,M_6,I_8,I_{10}$ in the Olive-Auffray minimal integrity basis of $\A$.   Find linear relations among these tenth degree powers and products.   Then we have two syzygy relations among these thirteen invariants as follows.
\begin{align}  \label{e1}
6J_2I_8 &= -I_2^2J_2K_4 - I_2^3L_4 + 3I_2I_4L_4 -3I_2J_4K_4 + 4J_2I_4K_4 + 2I_2^2 J_6 \nonumber \\
& + 3I_2J_2L_6 - 3L_4I_6 - 6I_4J_6 + 3J_4L_6 + 6K_4K_6,
\end{align}
and
\begin{align*}
&2I_2J_2K_6 + I_2^2 J_2 J_4 - I_2 J_4^2  + 2I_2K_4L_4 + 3 J_2 K_4^2  - 2 J_2 I_4  J_4 \\
&+  J_2^2 I_6 - 2I_2^2 M_6 - 12K_4 J_6 + 6 L_4L_6 + 6 I_4 M_6 - 3 J_4K_6 = 0.
\end{align*}
i.e.,
\begin{align} \label{e2}
(2I_2J_2 - 3 J_4 )K_6 & = - I_2^2 J_2 J_4 + I_2 J_4^2 - 2I_2 K_4L_4 - 3 J_2 K_4^2  + 2 J_2 I_4 J_4 \nonumber \\
&- J_2^2 I_6 + 2 I_2^2 M_6 + 12 K_4 J_6 - 6 L_4 L_6 - 6 I_4 M_6.
\end{align}

We first use the syzygy relation (\ref{e1}).
If $\uu = \0$, then $J_2 = u_iu_i = 0$, and the right hand side of
(\ref{e1}) is also equal to zero.  In this case, we have $I_8 = D_{ijk}D_{ij\ell}u_kD_{pq\ell}D_{pqr}v_r = 0$, where $v_p:=D_{ijk}D_{ij\ell}D_{k\ell p}$.
If $\uu \not = \0$, then $J_2 = u_iu_i \not = 0$.   By the syzygy relation (\ref{e1}), we have
\begin{align*}
I_8 &= -{1 \over 6} I_2^2 K_4 + {2 \over 3}I_4K_4 + {1 \over 2} I_2 L_6\\
& + {- I_2^3 L_4 + 3I_2I_4L_4 -3I_2J_4K_4  + 2 I_2^2 J_6
  - 3L_4I_6 - 6I_4J_6 + 3J_4L_6 + 6K_4K_6 \over 6 J_2}.
\end{align*}
Then $I_8$ is a function of $I_2,J_2,I_4,J_4,K_4,L_4,I_6,J_6,K_6,L_6,M_6,I_{10}$.

We now use the syzygy relation (\ref{e2}).    If $2I_2J_2 - 3 J_4 = 0$, by Proposition \ref{p3.1}, either $\D = \OOO$ or $\uu = \0$.  This implies that  $K_6 = 0$.  Note in this case, the right hand side of (\ref{e2}) also equal to zero.     If $2I_2J_2 - 3 J_4 \not = 0$, we have
\begin{align*}
& K_6  = \\ & {- I_2^2 J_2 J_4 + I_2 J_4^2 - 2I_2 K_4L_4 - 3 J_2 K_4^2  + 2 J_2 I_4 J_4 - J_2^2 I_6 + 2 I_2^2 M_6 + 12 K_4 J_6 - 6 L_4 L_6 - 6 I_4 M_6 \over 2I_2J_2 - 3 J_4}.
\end{align*}
This shows that $K_6$ is a function of $I_2,J_2,I_4,J_4,K_4,L_4,I_6,J_6,L_6,M_6,I_{10}$.

Hence, $\{I_2,J_2,I_4,J_4,K_4,L_4,I_6,J_6,L_6,M_6,I_{10}\}$ is a function basis of the third order three-dimensional symmetric tensor $\A$.
\end{proof}

\section{This Function Basis Is An Irreducible Function Basis}

To show that $\{I_2,J_2,I_4,J_4,K_4,L_4,I_6,J_6,L_6,M_6,I_{10}\}$ is an irreducible function basis of the third order three-dimensional symmetric tensor $\A$, we only need to show that
each of these eleven invariants is not a function of the other ten invariants.

To show that each of $K_4, L_4, J_6$ and $L_6$ is not a function of the ten other invariants in this function basis, we may use a tactics, which is stated in the following proposition.

\begin{Proposition} \label{p4.1}  We have the following four conclusions.

(a) If there is a third order three-dimensional tensor $\A$ such that $K_4 = L_4 = J_6 = 0$ but $L_6 \not = 0$, then $L_6$ is not a function of $I_2, J_2, I_4, J_4, K_4, L_4$, $I_6, J_6, M_6$ and $I_{10}$.

(b) If there is a third order three-dimensional tensor $\A$ such that $K_4 = L_4 = L_6 = 0$ but $J_6 \not = 0$, then $J_6$ is not a function of $I_2, J_2, I_4, J_4, K_4, L_4$, $I_6, L_6, M_6$ and $I_{10}$.

(c) If there is a third order three-dimensional tensor $\A$ such that $K_4 = J_6 = L_6 = 0$ but $L_4 \not = 0$, then $L_4$ is not a function of $I_2, J_2, I_4, J_4, K_4$, $I_6, J_6, L_6, M_6$ and $I_{10}$.

(d) If there is a third order three-dimensional tensor $\A$ such that $L_4 = J_6 = L_6 = 0$ but $K_4 \not = 0$, then $K_4$ is not a function of $I_2, J_2, I_4, J_4, L_4$, $I_6, J_6, L_6, M_6$ and $I_{10}$.
 \end{Proposition}
\begin{proof}
By the definition of invariants $I_2, J_2, I_4, J_4, K_4, L_4$, $I_6, J_6, L_6, M_6$ and $I_{10}$, if we keep $\D$ be unchanged but change $\uu$ to $-\uu$, then $I_2, J_2, I_4, J_4, I_6, M_6$ and $I_{10}$ are unchanged, but $K_4, L_4, J_6$ and $L_6$ change their signs.

We now prove conclusion (a).  If there is a third order three-dimensional tensor $\A$ such that $K_4 = L_4 = J_6 = 0$ but $L_6 \not = 0$, we may keep $\D$ be unchanged but change $\uu$ to $-\uu$, then $I_2, J_2, I_4, J_4, I_6, M_6$ and $I_{10}$ are unchanged,  $K_4, L_4$ and $J_6$ are still zeros, but $L_6$ changes its sign and value as it is not zero.   This implies that  $L_6$ is not a function of $I_2, J_2, I_4, J_4, K_4, L_4$, $I_6, J_6, M_6$ and $I_{10}$.    The other three conclusions (b), (c) and (d) can be proved similarly.
\end{proof}

We now present the main theorem of this section.

\begin{Theorem} \label{t4.2}
 The eleven invariant set $\{I_2,J_2,I_4,J_4,K_4,L_4,I_6,J_6,L_6,M_6,I_{10}\}$ is an irreducible function basis of the third order three-dimensional symmetric tensor $\A$.
\end{Theorem}
\begin{proof}  By Theorem \ref{t3.1}, $\{I_2,J_2,I_4,J_4,K_4,L_4,I_6,J_6,L_6,M_6,I_{10}\}$ is a function basis of $\A$.  It suffices to show that
each of these eleven invariants is not a function of the ten other invariants.

We divide the proof into three parts.

{\bf Part A.}  In this part, we show that each of $I_2, I_4, I_6, I_{10}$ and $J_2$ is not a function of the other ten invariants.   The first four invariants form an irreducible function basis of the symmetric and traceless tensor $\D$.  The fifth invariant $J_2$ forms an irreducible function basis of the vector $\uu$.   Using this property, we may prove that each of them is
not a function of the other ten invariants easily.

By Theorem \ref{t2.2}, $\{ I_2, I_4, I_6 , I_{10} \}$ is an irreducible function basis of $\D$.  This implies that each of these four invariants is not a function of the other three invariants.   Hence, each of these four invariants is not a function of the ten other invariants of $\{I_2,J_2,I_4,J_4,K_4,L_4,I_6,J_6,L_6,M_6,I_{10}\}$.

Let $\D = \OOO$, and $\uu \not = \uu'$ such that $u_iu_i \not = u'_iu'_i$.
Then $J_2$ takes two different values but the other ten invariants $I_2,I_4,J_4,K_4,L_4,I_6,J_6,L_6,M_6$ and $I_{10}$ are all zero.   This shows that $J_2$ is not a function of the ten other invariants $I_2,I_4,J_4,K_4,L_4,I_6,J_6,L_6,M_6$ and $I_{10}$.

{\bf Part B.}  In this part, we show that each of $K_4, L_4, J_6$ and $L_6$ is not a function of the ten other invariants.   We use Proposition \ref{p4.1} to realize this purpose.

We first show that $L_6$ is not a function of the other ten invariants.
Let $A_{111}, A_{112}, A_{113}, A_{122}$, $A_{123}, A_{133}, A_{222}, A_{223}, A_{233}$ and $A_{333}$ be the representatives of the components of $\A$.   If the values of these ten components are fixed, then the other components of $\A$ also fixed by symmetry.   Let $A_{111} = {3 \over 5}, A_{122} = {6 \over 5}, A_{133} = -{4 \over 5}, A_{223} = {1 \over 2}, A_{333} = -{1 \over 2}$, and $A_{112} = A_{113} = A_{123} = A_{222} = A_{233} = 0$,   Then we have $K_4 = L_4 = J_6 = 0$ and $L_6 = -2$, such that we may use  Proposition \ref{p4.1} (a).   The values of the other invariants are: $I_2 = 7, J_2 = 1, I_4 = {37 \over 2}$, $J_4 = 2$, $I_6 = 4$, $M_6 = 0, I_{10} = 4$.    By Proposition \ref{p4.1} (a), $L_6$ is not a function of $I_2, J_2, I_4, J_4, K_4, L_4$, $I_6, J_6, M_6$ and $I_{10}$.

 Then we show that $J_6$ is not a function of the other ten invariants.   Let
$$ A_{111} = {1 \over 6} \sqrt{{1\over 2} (149 - \sqrt{313})} - { 18 (-215 + 7 \sqrt{313})  \over
 5 \sqrt{8053043 - 308071 \sqrt{313}} } ,
$$
$$
 A_{112} =  { 121 (2963 - 103 \sqrt{313}) \over 10 (-215 + 7 \sqrt{313}) }  \sqrt{{298 - 2 \sqrt{313} \over
 648164815 - 26977811 \sqrt{313}}  } ,
$$
$$
 A_{113} = { 3966519 - 219867 \sqrt{313} \over 5 \sqrt{
 648164815 - 26977811 \sqrt{313}} (-215 + 7 \sqrt{313}) } ,
$$
$$
A_{122} = -{ 6 (-215 + 7 \sqrt{313}) \over  5 \sqrt{8053043 - 308071 \sqrt{313}} }, \qquad  A_{123} = 1 ,
$$
$$
A_{133} = -{1 \over 6} \sqrt{ {1\over 2} (149 - \sqrt{313})} - {6 (-215 + 7 \sqrt{313}) \over
 5 \sqrt{8053043 - 308071 \sqrt{313}} } ,
$$
$$
A_{222} = {363 (2963 - 103 \sqrt{313})  \over 10 (-215 + 7 \sqrt{313})}  \sqrt{{ 298 - 2 \sqrt{313}  \over
 648164815 - 26977811 \sqrt{313} }} ,
$$
$$
A_{223} = 1 + { 3966519 - 219867 \sqrt{313} \over
 5 \sqrt{648164815 - 26977811 \sqrt{313}} (-215 + 7 \sqrt{313}) } ,
$$
$$
A_{233} = { 121 (2963 - 103 \sqrt{313}) \over 10 (-215 + 7 \sqrt{313})} \sqrt{{298 - 2 \sqrt{313} \over
 648164815 - 26977811 \sqrt{313}}} ,
$$
$$
A_{333} = -1 + { 3 (3966519 - 219867 \sqrt{313})  \over
 5 \sqrt{648164815 - 26977811 \sqrt{313}} (-215 + 7 \sqrt{313})}.
$$
These values are solutions of $K_4 = L_4 = L_6 = 0$ and $J_6 \not = 0$.  Except that $A_{123}=1$, the approximate digit values of the other independent components are as follows:
$$
A_{111}=1.554, \  A_{112}=-0.1877, \  A_{113}=-0.01287, \  A_{122}=0.06780,
$$
$$
 A_{133}=-1.283, \  A_{222}=-0.5631, \  A_{223}=0.9871, \  A_{233}= -0.1877, \  A_{333} = -1.039.
$$
Then we have $K_4 = L_4 = L_6 = 0$ and $J_6 = 0.5112$, satisfying the condition of Proposition \ref{p4.1} (b).   The values of the other invariants are $I_2 = 17.29, J_2 = 1, I_4 = 132.6$, $J_4 = 2.547$, $I_6 = 83.81$, $M_6 = 0.1687$ and $I_{10} = -831$.   By Proposition \ref{p4.1} (b), $J_6$ is not a function of $I_2, J_2, I_4, J_4, K_4, L_4$, $I_6, L_6, M_6$ and $I_{10}$.

We now show that $L_4$ is not a function of the other ten invariants. Similarly, we can find a symmetric third order three-dimensional tensor $\A$ such that $K_4 = J_6 = L_6 = 0$ and $L_4 \neq 0$. To be specific, except that $A_{123}=1$,
the approximate digit values of the other independent components are as follows:
$$
A_{111}=1.0358, \  A_{112}=0.06373, \  A_{113}=-0.06357, \  A_{122}=1.8269,
$$
$$
 A_{133}=-1.9697, \  A_{222}=0.1912, \  A_{223}=0.9364, \  A_{233}= 0.06373, \  A_{333} = -1.1907.
$$
% Let $\A$ be decomposed to $\D$ and $\uu$ where
%the symmetric and traceless third order tensor $\D$ is given by
%$$
%D_{111}= 0.5,   D_{112}= 0,    D_{113}= 0,     D_{122}= 1.648312,   D_{123}= 1,   D_{222}= 0,   D_{223}= 1,
%$$
%and the vector $\uu$ is given by
%$$
%u_1 = 0.892996, u_2 = 0.318648, u_3 = -0.317838 .
%$$
Then we have $K_4 = J_6 = L_6 = 0$ and $L_4 = -0.3843$, satisfying the condition of Proposition \ref{p4.1} (c).   We also have $I_2 = 32.2465$, $J_2 = 1$, $I_4 = 394.69$, $J_4 = 9.1213$, $I_6 = 509.67$, $M_6 = 3.2506$ and $I_{10} = 17825.1$.
By Proposition \ref{p4.1} (c), $L_4$ is not a function of $I_2, J_2, I_4, J_4, K_4$, $I_6, J_6, L_6, M_6$ and $I_{10}$.

We further show that $K_4$ is not a function of the other ten invariants.
Let
$$ A_{111} = {3 \over 5 \sqrt{2}}, \quad
A_{112} =  {\sqrt{3} \over 10 }, \quad
A_{113} = { 1 \over 10 }, \quad
A_{122} = { 4 \sqrt{2}  \over 15 } - {1 \over \sqrt{3} }, \quad
A_{123} = {1\over 3} + {1\over \sqrt{6} } ,
$$
$$
A_{133} = -{\sqrt{2} \over 15} + {1 \over \sqrt{3} },   \quad
A_{222} = {3 \sqrt{3} \over 10},    \quad
A_{223} = -{9 \over 10 },           \quad
A_{233} = {\sqrt{3} \over 10 } ,    \quad
A_{333} = {13 \over 10}.
$$
Then we have $L_4 = J_6 = L_6 = 0$ and $K_4 = {8 \over 9}$, satisfying the condition of Proposition \ref{p4.1} (d).  We also have
$$ I_2 = 8, \ J_2 = {3 \over 2}, \ I_4 = {88 \over 3}, \ J_4 = {8 \over 3}, \ I_6 = {64 \over 9}, \ M_6 = {11 \over 9}, \ I_{10} = {11776 \over 729}. $$
By Proposition \ref{p4.1} (d), $K_4$ is not a function of $I_2, J_2, I_4, J_4, L_4$, $I_6, J_6, L_6, M_6$ and $I_{10}$.

{\bf Part C.}  In this part, we show that each of $M_6$ and $J_4$ is not a function of the ten other invariants.   We cannot use Proposition \ref{p4.1} here.   However, we may use another tactics.   We try to find a tensor $\A$ there such that $K_4 = L_4 = J_6 = L_6= 0$  to reduce the influence of these four invariants.   Then we change some components of $\A$ such that $K_4$, $L_4$, $J_6$ and $L_6$ keep to be zero, the value of $M_6$ or $J_4$ is changed and the values of the remaining other six invariants unchanged.

We first show that $M_6$ is not a function of the ten other invariants.    Let $u_1 = 5a, u_2=5b, u_3=5c$, $D_{123}=d$ and the other six independent components of $\D$ be zeros.   Let $a=b=0$ and $c=d=1$.  Then $I_2=6, J_2 =25, I_4=12, J_4=50$, $K_4=L_4=I_6=J_6=L_6=I_{10}=0$, and $M_6 =0$.  Let $a=b={\sqrt{2}\over 2}$, $c=0$ and $d=1$.  We still have $I_2=6, J_2 =25, I_4=12, J_4=50$, $K_4=L_4=I_6=J_6=L_6=I_{10}=0$, but $M_6 =625$.  Hence, $M_6$ is not a function of $I_2,J_2,I_4,J_4,K_4,L_4,I_6,J_6,L_6$ and $I_{10}$.

Finally, we show that $J_4$ is not a function of the ten other invariants.      Let $A_{111} = \frac{3}{5}\cos \theta,\  A_{112} = \frac{1}{5}\sin \theta, \ A_{113} = 0, \ A_{122} = \frac{1}{5}\cos \theta, \  A_{123} = 1, \ A_{133} = \frac{1}{5}\cos \theta, \ A_{222} =  \frac{3}{5}\sin \theta, \  A_{223} = 1, \  A_{233} = \frac{1}{5}\sin \theta$ and $A_{333} = -1$,   Then we really have $K_4 = L_4 = J_6 = L_6= 0$.  We also have $I_2 = 10, \ J_2 = 1, \ I_4 = 44, \ I_6 = 16, \ I_{10} = -64$ and
$$
J_4(\theta) = 2+4\cos \theta \sin \theta  + 2\sin^2 \theta, \qquad
M_6(\theta) =\sin^2 \theta (2\cos \theta+\sin^2 \theta).
$$
Clearly, $J_4(\frac{3}{4}\pi) = 1$, $M_6(\frac{3}{4}\pi)=\frac{1}{4}$, $M_6(0) = 0$ and $M_6(\frac{\pi}{4}) = \frac{9}{4}$. Since $M_6(\theta)$ is continuous in the interval $[0, \frac{\pi}{4}]$, there exists $\theta_0 \in (0, \frac{\pi}{4})$ such that $M_6(\theta_0)=M_6(\frac{3}{4}\pi)=\frac{1}{4}$. On the other hand, we have
$$
J_4'(\theta) = 4\cos(2\theta)+2\sin(2\theta) \geq 0, \qquad \forall \theta \in \left[0, \frac{\pi}{4}\right].
$$
It follows that $J_4(\theta_0) \geq J_4(0) =2 > J_4(\frac{3}{4}\pi) = 1$. Hence, $J_4$ is not a function of $I_2$, $J_2$, $I_4$, $K_4$, $L_4$, $I_6$, $J_6$, $L_6$, $M_6$ and $I_{10}$.

Combining the results of these three parts, each of these eleven invariants is not a function of the ten other invariants.
Therefore, this eleven invariant set
 $\{I_2,J_2,I_4,J_4,K_4,L_4,I_6,J_6,L_6,M_6,I_{10}\}$ is indeed an irreducible function basis of $\A$.

\end{proof}

Part A and the first part of Part C of this proof to show that each of $I_2, J_2, I_4, I_6, M_6, I_{10}$ is not a function of the other ten invariants may follow Theorem 3.1 of \cite{CQZ-18}.   For self-sufficiency and completeness of this paper, we give this part of the proof directly.   The organization of the proof to three parts also makes the proof an integral entity.

\section{Significance of This Result}

This result is significant to the further research of irreducible function bases of higher order tensors.   First, this is the first result on irreducible function bases of a third order three-dimensional symmetric tensor.  Second, there are still at least three
syzygy relations among these eleven invariants, see (\ref{e3}-\ref{e5}).  This shows that an irreducible function basis consisting of polynomial invariants may not be algebraically minimal in the sense that the basis consists of polynomial invariants and there is no algebraic relations in these invariants \cite{Sh-67}.     The second point is observed as there are still some syzygy relations among these eleven invariants.

Consider all possible sixteenth degree powers or products of the eleven invariants\\ $I_2,J_2,I_4,J_4,K_4,L_4,I_6,J_6,L_6,M_6,I_{10}$.  Find linear relations among these sixteenth-degree powers or products.   Then we have the following three syzygy relations among these eleven invariants as follows.

\begin{align} \label{e3}
&2 I_2^3 J_2^3 J_4 - 4 I_2 J_2^3 I_4 J_4 - 6 J_2^3 J_4 I_6 - 9 I_2^2 J_2^2 J_4^2 + 18 J_2^2 I_4 J_4^2 + 9 J_4^4 + 36 I_2 J_2 J_6^2  - 54 J_4 J_6^2  - 48 I_2 J_2^2 K_4 J_6 \nonumber \\
&  + 144 J_2 J_4 K_4 J_6  + 12 I_2 J_2^3 K_4^2 - 36 J_2^2 J_4 K_4^2 - 24 I_2^2 J_2 L_4 J_6  + 36 I_2 J_4 L_4 J_6  + 12 I_2^2 J_2^2 K_4 L_4   \nonumber \\
& - 18 I_2 J_2 J_4 K_4 L_4 - 18 J_4^2 K_4 L_4 + 6 I_2^3 J_2 L_4^2 - 6 I_2 J_2 I_4 L_4^2 - 9 I_2^2 J_4 L_4^2 + 9 I_4 J_4 L_4^2 - 36 J_2 J_4 L_4 L_6 \nonumber \\
& - 6 I_2^3 J_2^2 M_6 + 12 I_2 J_2^2 I_4 M_6 + 9 J_2^2 I_6 M_6 + 36 I_2^2 J_2 J_4 M_6 - 72 J_2 I_4 J_4 M_6 - 18 I_2 J_4^2 M_6 - 108 K_4 J_6 M_6 \nonumber \\
&   + 27 J_2 K_4^2 M_6 + 18 I_2 K_4 L_4 M_6 + 54 L_4 L_6 M_6 - 18 I_2^2 M_6^2 + 54 I_4 M_6^2 = 0,
\end{align}

\begin{align} \label{e4}
& {4 \over 9} I_2^3 J_2^3 K_4 + {2 \over 9} I_2^4 J_2^2 L_4 + {4 \over 3} I_2^3 J_2 J_4 L_4 - {8 \over 9} I_2 J_2^3 I_4 K_4 - {4 \over 9} I_2^2 J_2^2 I_4  L_4 - {4 \over 3} I_2^2 J_2^2 J_4 K_4 - 2 I_2^2 J_4^2 L_4 \nonumber \\
&  + 2 I_2^2 K_4 L_4^2 + 2 J_2^2 K_4^3 + 4 I_2 J_2 J_4^2 K_4  + 5 I_2 J_2 K_4^2 L_4 - 4 I_2 J_2 I_4 J_4 L_4 - {4 \over 3} I_2^3 J_2^2 J_6 + {2 \over 3} J_2^3 K_4 I_6 \nonumber \\
& + {1 \over 3} I_2 J_2^2 L_4 I_6 + {8 \over 3} I_2 J_2^2 I_4 J_6  + {4 \over 3} I_2^2 J_2 J_4 J_6 - 2 I_2^3 L_4 M_6 + J_2 J_4 L_4 I_6 - 16 I_2 K_4 L_4 J_6 - 14 J_2 K_4^2 J_6 \nonumber \\
&  + 6 I_2 L_4^2 L_6 + 4 J_2 K_4 L_4 L_6 + 6 I_2 I_4 L_4 M_6 - 2 I_2 J_4 K_4 M_6 + 4 J_2 I_4 K_4 M_6 + 4 I_2^2 J_6 M_6 - 2 J_2^2 I_6 J_6 \nonumber \\
&   - 4 I_2 J_2 L_6 M_6 - 12 I_4 J_6 M_6    + 6 J_4 L_6 M_6 + 24 K_4 J_6^2  - 12 L_4 J_6  L_6 - 4 J_4^3 K_4 + 4 I_4 J_4^2 L_4 - J_4 K_4^2  L_4 = 0
\end{align}

and

\begin{align} \label{e5}
&\frac{1}{18}I_2^5 J_2^3 - {2 \over 9} I_2^3 J_2^3 I_4 + {2 \over 9} I_2 J_2^3 I_4^2 + \frac{1}{12}I_2^2 J_2^3 I_6 - {1\over 6}J_2^3 I_4 I_6  - {1 \over 6} I_2^4 J_2^2 J_4 + {1 \over 3} I_2^2 J_2^2 I_4 J_4 + {1 \over 2} I_2 J_2^2 J_4 I_6  \nonumber \\
&  + {1\over 2} I_2^3 J_2 J_4^2 - I_2 J_2 I_4 J_4^2 - {3 \over 4} J_2 J_4^2 I_6 - {1 \over 2} I_2^2 J_4^3 + I_4 J_4^3  - I_2^2 J_2 K_4 J_6 + 2 J_2 I_4 K_4 J_6 + {1 \over 4} I_2^2 J_2^2 K_4^2  \nonumber \\
&  -  {1 \over 2} J_2^2 I_4 K_4^2 + {3 \over 2} I_2 J_2 J_4 K_4^2 - {9 \over 4} J_4^2 K_4^2 + {1\over 2} I_2^3 J_2 K_4 L_4 - I_2 J_2 I_4 K_4 L_4 - {1\over 2} I_2^2 J_4 K_4 L_4 + I_4 J_4 K_4 L_4   \nonumber \\
&  + 2 I_2 J_2 J_6 L_6 - 3 J_4 J_6 L_6 - 2 I_2 J_2^2 K_4 L_6  + 3 J_2 J_4 K_4 L_6 - {1 \over 2} I_2^2 J_2 L_4 L_6 - J_2 I_4 L_4 L_6 + {3\over 2} I_2 J_4 L_4 L_6   \nonumber \\
& - {1\over 6} I_2^4 J_2 M_6 + {5\over 6} I_2^2 J_2 I_4  M_6 - J_2 I_4^2 M_6 - I_2 J_2 I_6  M_6 + {3\over 2} J_4 I_6 M_6 =0.
\end{align}

As shown by Theorem \ref{t4.2}, these three syzygy relations do not imply any single-valued function relation of any of these eleven invariants, with respect to the ten other invariants.

The second point is meaningful to the further research of irreducible function bases of higher order tensors.  For example, for the nine invariant Smith-Bao minimal integrity basis of a fourth order three-dimensional symmetric and traceless tensor, there are five syzygy relations \cite{CQZ-18, Sh-67}.   These five syzygy relations are not so well-structured like (\ref{e1}) and (\ref{e2}), but even more complicated than (\ref{e3}-\ref{e5}).   However, it is still possible that  the nine invariant Smith-Bao minimal integrity basis is indeed an irreducible function basis of  a fourth order three-dimensional symmetric and traceless tensor, just like  the four invariant Smith-Bao minimal integrity basis is indeed an irreducible function basis of  a third order three-dimensional symmetric and traceless tensor, which was proved in \cite{CQZ-18}.

In Section 4, we show that the eleven invariant function basis is indeed an irreducible function basis, by showing that each of these eleven invariants is not a function of the ten other invariants.   This is the method proposed by Pennisi and Trovato \cite{PT-87}.   However, we divide the proof into three parts.  In Part A, we show that each of the five invariants $I_2, I_4, I_6$, $I_{10}$ and $J_2$, which form the irreducible function bases of the composition tensors $\D$ and $\uu$, is not a function of the ten other invariants.  In Part B, we use Proposition \ref{p4.1} to show that each of $K_4$, $L_4$, $L_4$ and $J_6$ is not a function of the ten other invariants.  In Part C, we use another tactics to show that each of the remaining two invariants $M_6$ and $J_4$ is not a function of the ten other invariants.   Such tactics may be also instructive for the further research of irreducible function bases of higher order tensors.

\end{document}